\documentclass[11pt,a4paper]{article}
\usepackage{vmargin}
\setmarginsrb{1.0in}{1.0in}{1.0in}{1.0in}{0mm}{0mm}{5mm}{5mm}
\usepackage{pifont}
\usepackage{amsmath}
\usepackage{amsfonts,amsmath,amssymb}
\usepackage{amsthm}
\usepackage{graphicx}
\usepackage{amsbsy}
\usepackage{amsmath}
\usepackage{color}
\usepackage{todonotes}
\usepackage{algorithm}
\usepackage{algpseudocode}
\usepackage{comment}
\usepackage{mathtools}
\usepackage{complexity}
\usepackage{subfigure}
\usepackage{hyperref}

\title{Selfish Network Creation with Non-Uniform Edge Cost}
\author{Ankit Chauhan\thanks{Algorithm Engineering Group, Hasso Plattner Institute Potsdam, Germany. \texttt{firstname.lastname@hpi.de}}  \and Pascal Lenzner\footnotemark[1] \and Anna Melnichenko\footnotemark[1] \and Louise Molitor\footnotemark[1]}

\date{~}

\newtheorem{theorem}{Theorem}
\newtheorem{lemma}{Lemma}
\newtheorem{corollary}{Corollary}

\newtheorem{observation}[theorem]{Observation}
\newtheorem{remark}{Remark}

\graphicspath{{pictures/}}

\begin{document}

\maketitle

\begin{abstract}
 \noindent Network creation games investigate complex networks from a game-theoretic point of view. 
 Based on the original model by Fabrikant et al.~[PODC'03] many variants have been introduced. However, almost all versions have the drawback that edges are treated uniformly, i.e. every edge has the same cost and that this common parameter heavily influences the outcomes and the analysis of these games. 
 
 We propose and analyze simple and natural parameter-free network creation games with non-uniform edge cost. Our models are inspired by social networks where the cost of forming a link is proportional to the popularity of the targeted node. Besides results on the complexity of computing a best response and on various properties of the sequential versions, we show that the most general version of our model has constant Price of Anarchy. 
 To the best of our knowledge, this is the first proof of a constant Price of Anarchy for any network creation game.

\end{abstract}

\section{Introduction}
Complex networks from the Internet to various (online) social networks have a huge impact on our lives and it is thus an important research challenge to understand these networks and the forces that shape them. The emergence of the Internet was one of the driving forces behind the rise of Algorithmic Game Theory~\cite{Pap01} and it has also kindled the interdisciplinary field of Network Science~\cite{Bar16}, which is devoted to analyzing and understanding real-world networks. Game-theoretic models for network creation lie in the intersection of both research directions and yield interesting insights into the structure and evolution of complex networks. In these models, agents are associated to nodes of a network and choose their neighbors selfishly to minimize their cost. 
Many such models have been proposed, most prominently the models of Jackson \& Wolinsky~\cite{JW96}, Bala \& Goyal~\cite{BG00} and Fabrikant et al.~\cite{Fab03}, but almost all of them treat edges equally, that is, they assume a fixed price for establishing any edge which is considered as a parameter of these games. This yields very simple models but has severe influence on the obtained equilibria and their properties. 

We take a radical departure from this assumption by proposing and analyzing a variant of the Network Creation Game~\cite{Fab03} in which the edges have non-uniform cost which solely depends on the structure of the network. In particular, the cost of an edge between agent $u$ and $v$ which is bought by agent $u$ is proportional to $v$'s degree in the network, i.e. edge costs are proportional to the degree of the other endpoint involved in the edge. Thus, we introduce individual prices for edges and at the same time we obtain a simple model which is parameter-free. 

Our model is inspired by social networks in which the nodes usually have very different levels of popularity which is proportional to their degree. In such networks connecting to a celebrity usually is expensive. Hence, we assume that establishing a link to a popular high degree node has higher cost than connecting to an unimportant low degree node. Moreover, in social networks links are formed mostly locally, e.g. between agents with a common neighbor, and it rarely happens that links are removed, on the contrary, such networks tend to get denser over time~\cite{LKF05}. This motivates two other extensions of our model which consider locality and edge additions only.   

\subsection{Model and Notation}

Throughout the paper we will consider unweighted undirected networks $G = (V,E)$, where $V$ is the set of nodes and $E$ is the set of edges of $G$. Since edges are unweighted, the distance $d_G(u,v)$ between two nodes $u,v$ in $G$ is the number of edges on a shortest path between $u$ and $v$. For a given node $u$ in a network $G$ let $N_k(u)$ be the set of nodes which are at distance at most $k$ from node $u$ in $G$ and let $B_k(u)$ be the set of nodes which are at exactly distance $k$ from node $u$ (the distance-$k$ ball around $u$). We denote the diameter of a network $G$ by $D(G)$, the degree of node $u$ in $G$, which is the number of edges incident to $u$, by $deg_G(u)$. 
We will omit the reference to $G$ whenever it is clear from the context.  

We investigate a natural variant of the well-known Network Creation Game (NCG) by Fabrikant et al.~\cite{Fab03} which we call the \emph{degree price network creation game (degNCG)}. In a NCG the selfish agents correspond to nodes in a network and the strategies of all agents determine which edges are present. In particular, the strategy $S_u$ of an agent $u$ is any subset of $V$, where $v \in S_u$ corresponds to agent $u$ owning the undirected edge $\{u,v\}$. For $v \in S_u$ we will say that agent $u$ buys the edge $\{u,v\}$.
Any strategy vector $\mathbf{s}$ which specifies a strategy for each agent then induces the network $G(\mathbf{s})$, where $$G(\mathbf{s}) = \left(V, \bigcup_{u\in V} \bigcup_{v \in S_u} \{u,v\}\right).$$
Here we assume that $G(\mathbf{s})$ does not contain multi-edges, which implies that every edge has exactly one owner. Since edge-ownership is costly (see below) this assumption trivially holds in any equilibrium network. Moreover, any network $G$ together with an ownership function, which assigns a unique owner for every edge, determines the corresponding strategy vector. Hence, we use strategy vectors and networks interchangeably and we assume that the owner of every edge is known. In our illustrations we indicate edge ownership by directing the edges away from their owner. We will draw undirected edges if the ownership does not matter.

The cost function of agent $u$ in network $G(\mathbf{s})$ consists of the sum of edge costs for all edges owned by agent $u$ and the distance cost, which is defined as the sum of the distances to all other nodes in the network if it is connected and $\infty$ otherwise. The main novel feature which distinguishes the degNCG from the NCG is that each edge has an individual price which is proportional to the degree of the endpoint which is not the owner. That is, if agent $u$ buys the edge $\{u,v\}$ then $u$'s cost for this edge is proportional to node $v$'s degree. For simplicity we will mostly consider the case where the price of edge $\{u,v\}$ for agent $u$ is exactly $v$'s degree without counting edge $\{u,v\}$.
Thus the cost of agent $u$ in network $G(\mathbf{s})$ is     
$$cost_u(G(\mathbf{s})) = \sum_{v \in S_u}(\deg_{G(\mathbf{s})}(v)-1) + dist_{G(\mathbf{s})}(u)\enspace.$$
Note that in contrast to the NCG our variant of the model does not depend on any parameter and the rather unrealistic assumption that all edges have the same price is replaced with the assumption that buying edges to well-connected nodes is more expensive than connecting to low degree nodes.  

Given any network $G(\mathbf{s})$, where agent $u$ has chosen strategy $S_u$. We say that $S_u$ is a best response strategy of agent $u$, if $S_u$ minimizes agent $u$'s cost, given that the strategies of all other agents are fixed. We say that a network $G(\mathbf{s})$ is in pure Nash equilibrium (NE), if no agent can strictly decrease her cost by unilaterally replacing her current strategy with some other strategy. That is, a network $G(\mathbf{s})$ is in NE if all agents play a best response strategy.

Observe that in the degNCG we assume that agents can buy edges to every node in the network. Especially in modeling large social networks, this assumption seems unrealistic. To address this, we also consider a restricted version of the model which includes locality, i.e. where only edges to nodes in distance at most $k$, for some fixed $k\geq 2$, may be bought. We call this version the \emph{$k$-local degNCG} (deg$k$NCG) and its pure Nash equilibria are called $k$-local NE ($k$NE). We will mostly consider the case strongest version where $k=2$. 

We measure the quality of a network $G(\mathbf{s})$ by its \emph{social cost}, which is simply the sum over all agents' costs, i.e. $cost(G(\mathbf{s})) = \sum_{u\in V}cost_u(G(\mathbf{s}))$. Let $\text{\emph{worst}}_n$ and $\text{\emph{best}}_n$ denote the social cost of a ($k$)NE network on $n$ nodes which has the highest and lowest social cost, respectively. Moreover, let $\text{\emph{opt}}_n$ be the minimum social cost of any network on $n$ nodes. We measure the deterioration due to selfishness by the Price of Anarchy (PoA) which is the maximum over all $n$ of the ratio $\frac{\text{\emph{worst}}_n}{\text{\emph{opt}}_n}$. Moreover, the more optimistic Price of Stability (PoS) is the maximum over all $n$ of the ratio $\frac{\text{\emph{best}}_n}{\text{\emph{opt}}_n}$. 

The use case of modeling social networks indicates another interesting version of the degNCG, which we call the \emph{degree price add-only game (degAOG)} and its $k$-local version deg$k$AOG. In these games, agents can only add edges to the network whereas removing edges is impossible. This mirrors social networks where an edge means that both agents know each other.

\subsection{Related Work}
Our model is a variant of the well-known Network Creation Game (NCG) proposed by Fabrikant et al.~\cite{Fab03}. The main difference to our model is that in~\cite{Fab03} it is assumed that every edge has price $\alpha>0$, where $\alpha$ is some fixed parameter of the game. This parameter heavily influences the game, e.g. the structure of the equilibrium networks changes from a clique for very low $\alpha$ to trees for high $\alpha$. For different regimes of $\alpha$ different proof techniques yield a constant PoA~\cite{Fab03,Al14,De07,MS10,MMM13,GHLL16} but it is still open whether the PoA is constant for all $\alpha$. In particular, constant upper bounds on the PoA are known for $\alpha < n^{1-\varepsilon}$, for any fixed $\varepsilon >\frac{1}{\log n}$~\cite{De07}, and if $\alpha > 65n$~\cite{MMM13}. The best general upper bound is $2^{\mathcal{O}\sqrt{\log n}}$~\cite{De07}. 
The dynamics of the NCG have been studied in~\cite{KL13} where it was shown that there cannot exist a generalized ordinal potential function for the NCG. Also the complexity of computing a best response has been studied and its \NP-hardness was shown in~\cite{Fab03}. If agents resort to simple strategy changes then computing a best response can trivially be done in polynomial time and the obtained equilibria approximate Nash equilibria well~\cite{Len12}.

Removing the parameter $\alpha$ by restricting the agents to edge swaps was proposed and analyzed in~\cite{MS12,ADHL13}. The obtained results are similar, e.g. the best known upper bound on the PoA is $2^{\mathcal{O}\sqrt{\log n}}$, there cannot exist a potential function~\cite{L11} and computing a best response is \NP-hard. However, allowing only swaps leads to the unnatural effects that the number of edges cannot change and that the sequential version heavily depends on the initial network.     

Several versions for augmenting the NCG with locality have been proposed and analyzed recently. It was shown that the PoA may deteriorate heavily if agents only know their local neighborhood or only a shortest path tree of the network~\cite{Bil14local,Bil14traceroute}. In contrast, a global view with a restriction to only local edge-purchases yields only a moderate increase of the PoA~\cite{CL15}. 

The idea of having nodes with different popularity was also discussed in the so called celebrity games~\cite{ABDMS16,AM16}. There, nodes have a given popularity and agents buy fixed-price edges to get highly popular nodes within some given distance bound. Hence, this model differs heavily from our model.

To the best of our knowledge, there are only two related papers which analyze a variant of the NCG with non-uniform edge price. In~\cite{CMadH14} agents can buy edges of different quality which corresponds to their length and the edge price depends on the edge quality. Distances are then measured in the induced weighted network. Closer to our model is~\cite{MMO14} where heterogeneous agents, important and unimportant ones, are considered and both classes of agents have different edge costs. Here, links are formed with bilateral agreement~\cite{JW96,CP05} and important nodes have a higher weight, which increases their attractiveness.        

\subsection{Our Contribution}
We introduce and analyze the first parameter-free variants of Network Creation Games~\cite{Fab03} which incorporate non-uniform edge cost. In almost all known versions the outcomes of the games and their analysis heavily depend on the edge cost parameter $\alpha$. We depart from this by assuming that the cost of an edge solely depends on structural properties of the network, in particular, on the degree of the endpoint to which the edge is bought. Essentially, our models incorporate that the cost of an edge is proportional to the popularity of the node to which it connects. This appears to be a realistic feature, e.g. for modeling social networks.  

On the first glance, introducing non-uniform edge cost seems to be detrimental to the analysis of the model. However, in contrast to this, we give a simple proof that the PoA of the degNCG is actually constant. To the best of our knowledge, our model is the first version of the NCG for which this strong statement could be established. A constant PoA is widely conjectured for the original NCG~\cite{Fab03} but its proof is despite serious efforts of the community still an open question. Besides this strongest possible bound on the PoA, which we also generalize to arbitrary linear functions of a node's degree and to the $4$-local version, we prove a PoA upper bound of $\mathcal{O}(\sqrt{n})$ for the deg$2$NCG, where agents are restricted to act within their $2$-neighborhood and we show for this version that computing a best response strategy is \NP-hard. Moreover, we investigate the dynamic properties of the deg($2$)NCG and prove that improving response dynamics may not converge to an equilibrium, that is, there cannot exist a generalized ordinal potential function.

We contrast these negative convergence results by analyzing a version where agents can only add edges, i.e. the deg($2$)AOG, where convergence of the sequential version is trivially guaranteed, and by analyzing the speed of convergence for different agent activation schemes. The restriction to only edge additions has severe impact on the PoA, yielding a $\Theta(n)$ bound, but we show that the impact on the social cost is low, if round-robin dynamics starting from a path are considered, where agents buy their best possible single edge in each step.

\section{Hardness}\label{sec:hardness}
In this section we investigate the computational hardness of computing a cost minimizing strategy, i.e. a best response, in the deg$2$NCG and in the deg$2$AOG.

\begin{lemma}\cite{AK00}
\textsc{Dominating Set} in $q$-regular graphs is \NP-hard.
\end{lemma}
\noindent To prove the \NP-hardness of computing a best response in the deg$2$AOG and the deg$2$NCG we first define a variant of \textsc{Set Cover}~\cite{DKH11}, called \textsc{Exact-$q$ Set Cover} and prove that it is \NP-hard via a reduction from \textsc{Dominating Set} in $q$-regular graphs. The  \textsc{Exact-$q$ Set Cover} problem is defined as follows:  Given a universal set $U=\{1,2,\dots,n\}$ and a collection of sets $\mathcal{A}=\{A_1,\cdots A_{\ell}\}$ where $\forall A_i\in \mathcal{A}, A_i\subseteq U$ and $|A_i|=q$ with $q\geq 4$; the task is to find the minimum subset of $\mathcal{A}$ which covers $U$.
\begin{lemma}
\textsc{Exact-$q$ Set Cover} is \NP-hard.
\end{lemma}
\begin{proof}
We give a polynomial time reduction from \textsc{Dominating Set} in $q$-regular graphs to \textsc{Exact-$(q+1)$-Set Cover}. Given a $q$-regular graph $G=(V,E)$ with $V=\{1,\cdots, n\}$, construct an instance of \textsc{Exact-$(q+1)$-Set Cover} as follows: The universe is $U=V$ and the collection of sets is $\mathcal{A}=\{A_1,\cdots A_{n}\}$ such that $A_i$ consists of vertex $i$ and the vertices adjacent to vertex $i$. Since $G$ is $q$-regular, we have that $|A_i|=q+1$. Now $S$ is an optimal solution for \textsc{Exact-$(q+1)$-Set Cover} for the instance $(U,\mathcal{A})$ if and only if $R = \{i\mid A_i\in S\}$ is a minimum dominating set for the graph $G$.
\end{proof}
\noindent We now use the \NP-hardness of \textsc{Exact-$q$ Set Cover} to prove that the problem of computing a best response in the deg$2$NCG and in the deg$2$AOG is \NP-hard.

\begin{theorem}\label{thm:np2}
\item Computing the best response in the deg2NCG and the deg2AOG is \NP-hard. 
\end{theorem}

\begin{proof} We prove the theorem by giving a polynomial time reduction from the \NP-hard \textsc{Exact-$q$ Set Cover} problem to the problem of computing a best response of an agent $u$ in an instance of the deg2NCG. Consider an instance $I$  of the \textsc{Exact-$q$ Set Cover} problem with universal set $U=\{1,2,\dots,n\}$ and a collection of sets $\mathcal{A}=\{A_1,\cdots A_{\ell}\}$, where $A_i \subseteq U$ for all $A_i \in \mathcal{A}$ and $|A_i|=q$. We create a corresponding instance of the best response problem in the deg2NCG, where the network $G=(V,E)$ is defined as follows: $$V=\{v_1,\cdots, v_n\}\cup \left\{p_1^1,\dots , p_1^{q+1},\dots ,p_n^{1}\dots p_n^{q+1} \right\}\cup\{a_1,\cdots, a_{\ell}\} \cup \{x\}\cup \{u\},$$ where each $v_i$ corresponds to $i\in U$ and each $a_i$ corresponds to $A_i\in \mathcal{A}$. The edge set $E$ is defined as follows: 
\begin{align*}E =& \left\{\{v_i,a_j\} \mid i \in A_j\right\} \cup \left\{\{v_i,p_i^r\} \mid i\in U \text{ and } 1\leq r \leq q+1\right\}\\ &\cup \left\{\{a_j,x\} \mid a_j \in V\right\} \cup \{\{u,x\}\}\enspace.\end{align*} Moreover, we assume that the edge $\{u,x\}$ is owned by agent $x$. See Fig.~\ref{fig:NPhard_degLAOG}.
\begin{figure}[ht]
\center{\includegraphics[width=9cm]{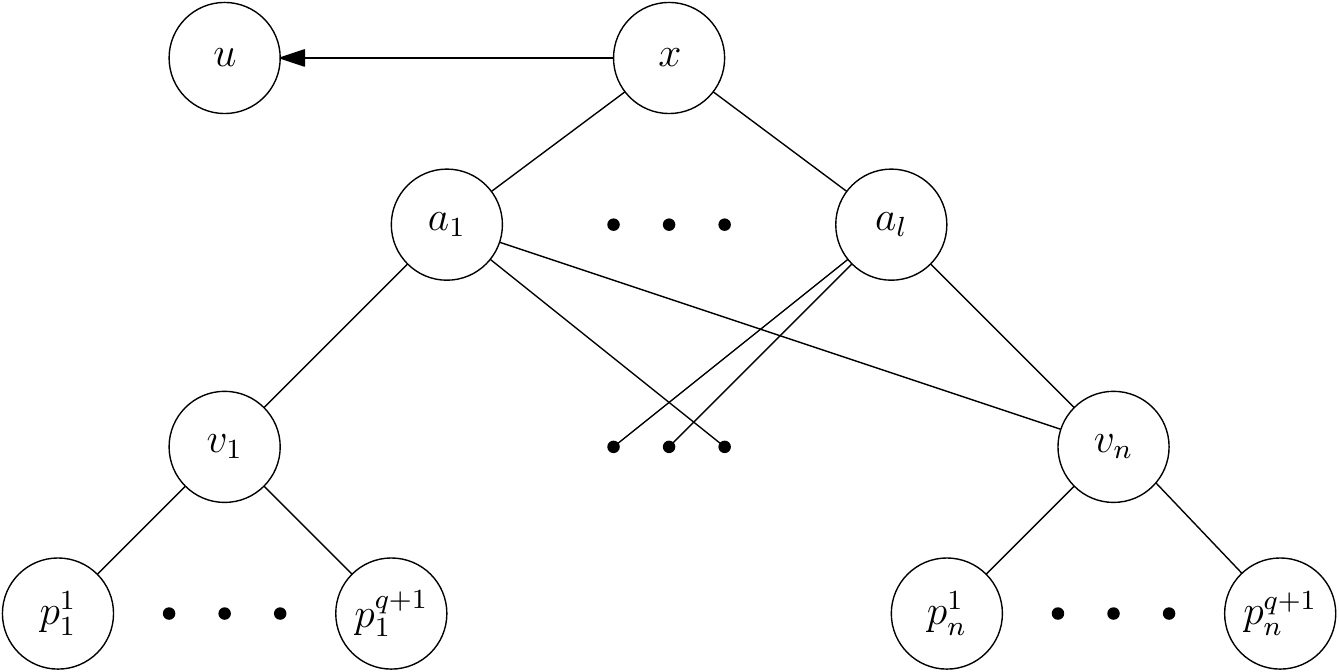}}
\caption{Illustration of the construction used in the reduction.}\label{fig:NPhard_degLAOG}
\end{figure} 

\noindent Note that the degree of every node in $B_2(u)$ is exactly $q+1$ and in $B_3(u)$ at least $q+2$.
Now we observe that if agent $u$ is playing her best response this means agent $u$ has to create a minimum number of edges to nodes in $B_2(u)$ such that every node $v_i$ is in distance $2$ from $u$. This is true since otherwise there exists a node $a_j\in B_2(u)$ to which buying an edge costs $q+1$ and the resulting improvement in distance cost is at least $q+3$ which contradicts the fact that $u$ is playing her best response. Moreover, since $u$ does not own any edges, no deletions or edge swaps are possible. 

Let $S_u^* \subseteq B_2(u)$ denote agent $u$'s best response. Since $S_u^*$ has minimum size such that all $v_i$ nodes are in distance $2$ to agent $u$, the union of the sets $A_j\in \mathcal{A}$ corresponding to the chosen nodes $a_j \in S_u^*$ forms the minimum size solution for the \textsc{Exact-$q$ Set Cover} instance $I$.

Since in the above reduction the best response of agent $u$ consists of only buying edges, this implies \NP-hardness for the deg$2$AOG.
\end{proof}

\section{Analysis of Equilibria}\label{sec:Eq}
We start with the most fundamental statement about equilibria which is their existence. We use the center sponsored spanning star $S_n$, see Fig.~\ref{star}, for the proof and provide some other examples of NE and $2$NE networks in Fig.~\ref{fig:NEsamples}.  
\begin{theorem}\label{thm:existence}
The star $S_n$ is a (k)NE for the deg(k)NCG and the deg(k)AOG for any $k$. 
\end{theorem}
\begin{proof}
 In the center sponsored spanning star $S_n$ the center agent cannot delete or swap any edge since this would disconnect the network. Since the center already has bought the maximum number of edges, no edge purchases are possible. Moreover, no leaf agent can profit from buying any number of edges because only edges to other leaves can be bought, which is a $2$-local move. Such edges have cost of at $1$ which equals the maximum possible distance improvement. Thus no agent has an improving move for any $k$ which implies that $S_n$ is in ($k$)NE.
\end{proof}

\begin{figure}[!h]
\subfigure[The NE network $S_n$]
{\includegraphics[width=0.16\textwidth]{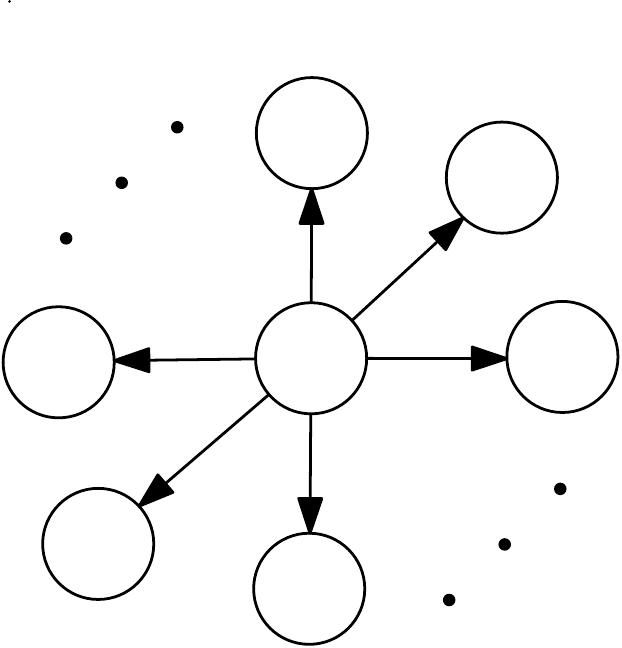}\label{star}}~~\hfill
\subfigure[A NE in the degAOG/degNCG with $D=3$. ]{\includegraphics[width=0.22\textwidth]{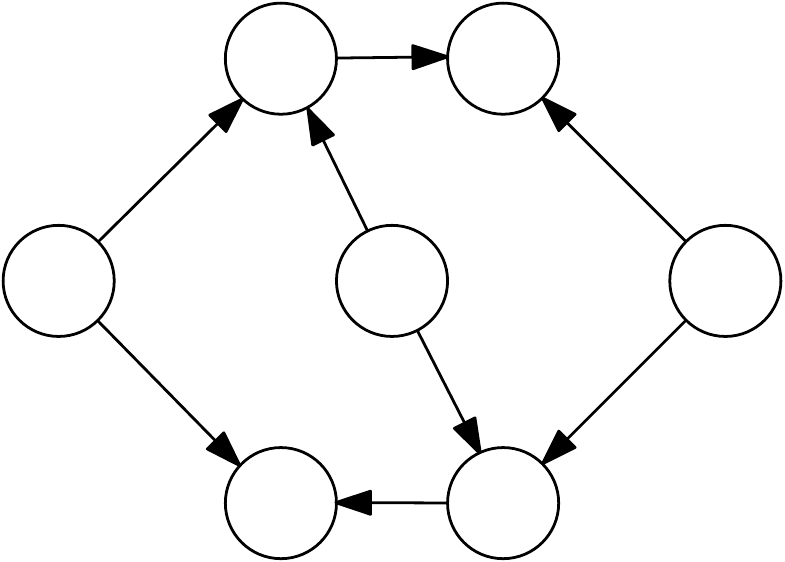}\label{D3}}~~\hfill
\subfigure[A $2$NE in the deg$2$NCG with $D = 4$.]{\includegraphics[width=0.26\textwidth]{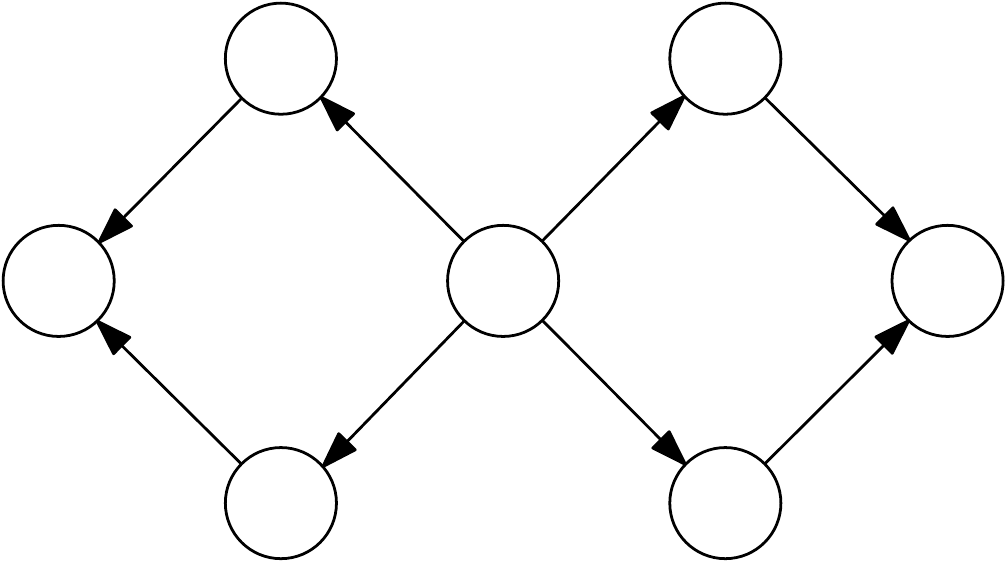}\label{2-vertexconnected}}~~\hfill
\subfigure[A $2$NE in the deg$2$AOG with $D=5$.]{\includegraphics[width=0.29\textwidth]{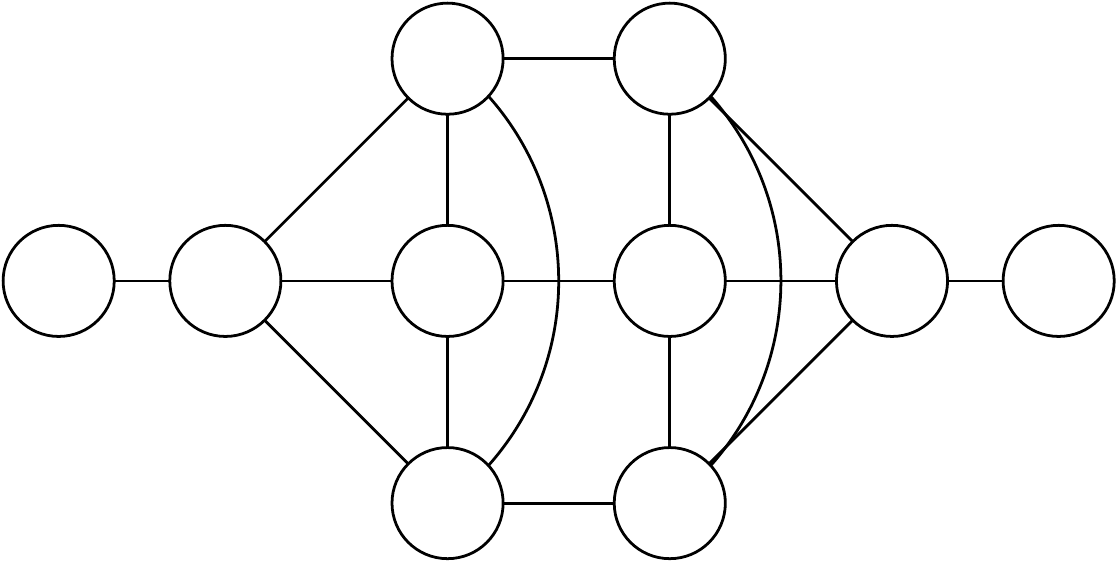}\label{D5_degLAOG}}
\caption{Examples of NE and $2$NE networks}
\label{fig:NEsamples}
\end{figure}

\subsection{Bounding the Diameter of Equilibrium Networks}\label{sec:Diam}
We investigate the diameter of ($2$)NE networks. Analogously to the original NCG~\cite{Fab03}, bounding the diameter plays an important role in bounding the PoA.
\begin{theorem}\label{thm:generaldiam}
Consider variants of the degAOG and the degNCG where the price of any edge $\{u,v\}$ bought by agent $u$ is any linear function of $v$'s degree in $G$, that is, $price_u(\{u,v\})=\beta\cdot~ deg_{G(s)}(v)+\gamma$, where $\beta, \gamma \in \mathbb{R}$. Then the diameter of any NE network in the degAOG and the degNCG is constant.
\end{theorem}
\begin{proof}
We consider a NE network $G = (V,E)$ and assume that the diameter $D$ of $G$ is at least 4. Then there exist nodes $a,b \in V$, such that $d_G(a,b)=D$. Therefore, the distance cost of a agent $a$ in $G$ is at least $D+|B_1(b)|(D-1)+|N_2(a)|$. Thus, if agent $a$ buys the edge $\{a,b\}$ then this improves agent $a$'s distance cost by at least $D-1+|B_1(b)|(D-3)$. Since the network $G$ is in NE, the distance cost improvement must be less than agent $u$'s cost for buying the edge $\{a,b\}$:
\begin{align*}
D-1+|B_1(b)|(D-3) &\leq \beta\cdot deg_{G}(b)+\gamma\\
\iff D-1+(D-3)\cdot deg_{G}(b) &\leq \beta\cdot deg_{G}(b)+\gamma\enspace.
\end{align*} 
Solving for $D$ under the assumption $deg_G(b) \geq 1$ yields $$D \leq \frac{(\beta+3)deg_G(b)+\gamma+1}{deg_G(b)+1} < \beta +3 + \frac{\gamma+1}{deg_G(b)+1} \in \mathcal{O}(1).$$
\end{proof}
\noindent Using $\beta=1$ and $\gamma = -1$ yields the edge price for our version of the degNCG and the degAOG. This, and the NE example in Fig.~\ref{D3} yields the following: 
\begin{corollary}\label{thm:diameter}
The diameter of any NE network in the degAOG and degNCG is at most 3 and this upper bound is tight.
\end{corollary}
\noindent Since in the proof of Theorem~\ref{thm:generaldiam} in the case of $\beta = 1$ and $\gamma = -1$ buying an edge to a node in distance 4 suffices, we get the following statement.
\begin{corollary}\label{cor:diam4local}
Any 4NE network has diameter at most $3$. 
\end{corollary}
\noindent Note that the examples in Fig.~\ref{2-vertexconnected} and Fig.~\ref{D5_degLAOG} show that the diameter in the $2$-local version, i.e. in the deg$2$NCG and the deg$2$AOG, can exceed $3$. We prove a higher upper bound on the diameter for the $2$-local versions.

\begin{theorem}\label{thm_diam_LNE}
The diameter of any 2NE network is in $\mathcal{O}(\sqrt{n})$.
\end{theorem}
\begin{proof}
Consider a $2$NE network $G=(V,E)$ with $|V|=n$ and let $D$ denote its diameter. Consider two nodes $a,b\in V$ such that $d_G(a,b)=D$ and a shortest-path tree $T_a=(V,E_a)$ which is rooted at node $a$ (see Fig.~\ref{s-p_tree}).
\begin{figure}[h!]
\center{\includegraphics[width=7cm]{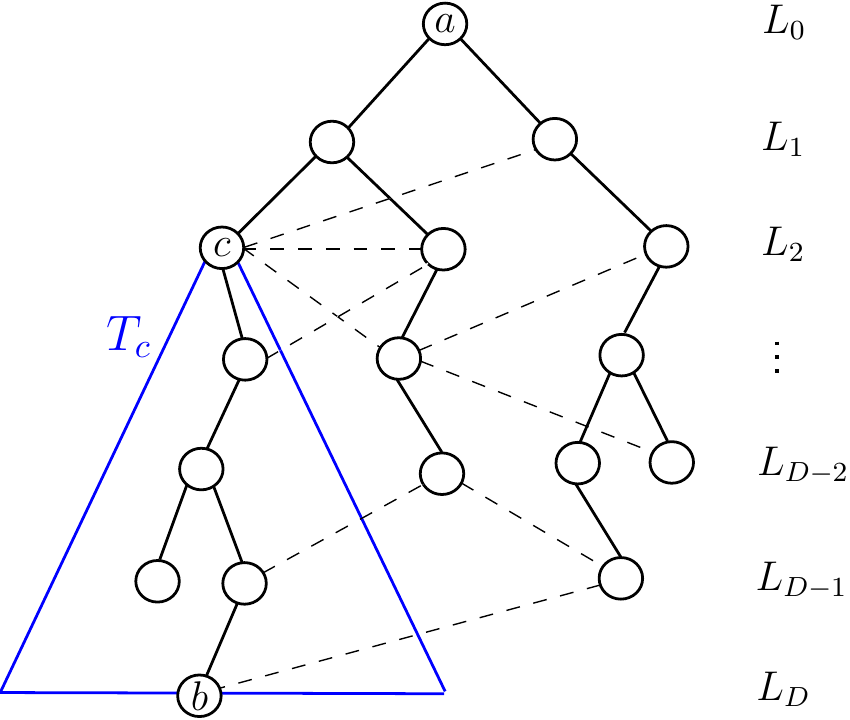}}
\caption{The shortest-path tree $T_a$. Dashed lines denote edges of $G$ which are not in the tree, i.e. the non-tree edges.}\label{s-p_tree}
\end{figure}

\noindent The height of $T_a$ is $D$ and there must be a subtree $T_c$ which contains node $b$ and which has node $c$ as root, where $c$ is chosen such that $d_G(a,c)=2$ and $c$ belongs to the path from $a$ to $b$ in $T_a$. Since the height of $T_c$ is $D-2$ it follows that the number of nodes in $T_c$ must be at least $D-1$. Let $|T_x|$ denote the number of nodes in the subtree of $T_a$ rooted at node $x$. Hence, we have $|T_c| \geq D-1$.

Note that if agent $a$ buys any edge $\{a,x\}$ in network $G$ then this improves $a$'s distance cost by at least $|T_x|$. 
Since $G$ is in $2$NE, we know that buying the edge $\{a,c\}$ is not an improving move for agent $a$ which implies that $|T_c|$ is at most the cost of the edge $\{a,c\}$ which is equal to $deg_G(c)$. Since $|T_c|\geq D-1$ it follows that $deg_G(c)\geq D-1$.  

Let $L_i$ denote the set of nodes which are in distance $i$ from the root $a$ in the tree $T_a$. For example $L_0=\{a\}, c\in L_2$ and $b\in L_D$. 
Thus, we have $D-1\leq deg_G(c)\leq |L_1|+(|L_2|-1)+|L_3|$. 

Analogously, since $G$ is in $2$NE, we have that no agent $v_i$ in layer $L_i$ on the $c-b$ path in $T_a$ can decrease her cost by buying an edge to a node in layer $L_{i+2}$ which is a neighbor of a neighbor in $T_a$. With analogue reasoning as above we get $D-(i-1)\leq deg_G(v_i)\leq |L_{i-1}|+(|L_i|-1)+|L_{i+1}|$.

Note that not only agents from lower layers cannot improve by buying edges towards nodes in upper layers but also agents from upper layers cannot improve by buying edges towards nodes in lower layers. Thus we have $$D-(i-1)\leq deg_G(v_i)\leq |L_{i-1}|+(|L_i|-1)+|L_{i+1}|$$ and $$D-(i-1)\leq deg_G(v_{D-i})\leq |L_{D-i-1}|+(|L_{D-i}|-1)+|L_{D-i+1}|$$  for any $2\leq i \leq \left\lfloor\frac{D}{2}\right\rfloor -1$. 
Summing up all inequalities yields:
$$2\sum_{i=2}^{\left\lfloor\frac{D}{2}\right\rfloor-1}\big(D-(i-1)\big)\leq 3\left(\sum_{i=1}^{D}{|L_i|}-(D-1)\right).$$ For the left side we have 
$$\frac{3D^2}{4}-4D-3 < \left(\left\lfloor\frac{D}{2}\right\rfloor-2\right)\left(2D+1 - \left\lfloor\frac{D}{2}\right\rfloor\right) = 2\sum_{i=2}^{\left\lfloor\frac{D}{2}\right\rfloor-1}\big(D-(i-1)\big)$$
and the right side gives $3\left(\sum_{i=1}^{D}{|L_i|}-(D-1)\right) \leq 3n-3D+3$, which yields 
$$\frac{3D^2}{4}-4D-3< 3n-3D+3 \Rightarrow D < \frac{2}{3}\left(1+\sqrt{9n+19}\right) \in \mathcal{O}(\sqrt{n}). $$
\end{proof}
\subsection{Price of Stability} \label{sec:PoS}
For analyzing the Price of Stability, we have to investigate the network which has the minimum possible social cost. 
\begin{lemma}\label{lem:optdeg}
The center sponsored spanning star $S_n$ is an optimal solution of the deg(k)NCG and the deg(k)AOG for any $k$.   
\end{lemma}
\begin{proof}
Consider an optimal network $G$ with $m$ edges and $n$ nodes. As $G$ has to be connected  we have $m\geq n-1$. Now all the pairs which are not connected by an edge are in distance of at least $2$ and there are $n(n-1)-2m$ many such pairs. Adding the remaining $2m$ pairs with distance $1$ yield distance cost of $2(n(n-1)-2m)+2m = 2n(n-1)-2m$  which is also the lower bound on the distance cost of any graph. Since $S_n$ has diameter 2 and all edges have costs of $0$ because every leaf node has degree 1, the social cost of the center sponsored spanning star $S_n$ meets the above lower bound.        
\end{proof}
\noindent We have shown in the proof of Theorem~\ref{thm:existence} that the center sponsored spanning star $S_n$ is in ($k$)NE for any $k$. With Lemma~\ref{lem:optdeg} this yields the following for $k\geq 2$.
\begin{corollary}\label{thm:PoS}
The Price of Stability of deg(k)NCG and the deg(k)AOG is 1.
\end{corollary}

\subsection{Price of Anarchy} \label{sec:PoA}
For investigating the quality of the equilibria of our games, we first adapt an important lemma by Fabrikant et al.~\cite{Fab03} to our setting. 
\begin{lemma}\label{lem:Edgecost}
If a (k)NE network $G$ in the deg(k)NCG has diameter $D$, then its social cost is at most $\mathcal{O}(D)$ times the minimum possible social cost. 
\end{lemma}
\begin{proof}
The minimum possible social cost is at least $n^2-1 \in \Omega(n^2)$, as the network is connected and every pair of nodes is at least in distance $1$. To bound the social cost of of the ($k$)NE $G$, we bound the social distance cost and social edge cost separately. A trivial upper bound for the social distance cost is $n^2D$, since $G$ has diameter $D$.

For bounding the social edge cost we first consider bridges of $G$, which are edges whose removal will disconnect $G$. There are at most $n-1$ bridges, so the total edge cost of all bridges is at most $\Delta(n-1)\in \mathcal{O}(n^2)$ with $\Delta = \max_{v \in V}deg_G(v)$. Now we will argue that the cost of all non-bridges bought by any agent $u$ is in $\mathcal{O}(nD)$, which implies that the total edge cost of all edges is in $\mathcal{O}(n^2D)$. This yields an upper bound on the social cost of $\mathcal{O}(n^2D+n^2+n^2D)= \mathcal{O}(n^2D)$, completing the proof.    

Consider an agent $u$ and fix agent $u$'s shortest path tree $T_u$, that is, we fix a shortest path from $u$ to all other nodes in $G$. Let $\{u,v\}$ be any non-bridge edge bought by agent $u$. Let $R_v$ be the set of nodes $w$, where the shortest path from $u$ to $w$ in $T_u$ contains the node $v$. 

We first argue that the distance between $u$ and $v$ is at most $2D$ if agent $u$ would remove the edge $\{u,v\}$. Note that removing $\{u,v\}$ cannot disconnect the network, since $\{u,v\}$ is not a bridge.
Let $\{x,w\}$ be the edge on some shortest path from $u$ to $v$ in $G=(V,E\setminus{\{u,v\}})$  where $x\notin R_v$ and $w \in R_v$. As the diameter of $G$ is $D$ and since $x \notin R_v$ there must be a path of length at most $D$ between $u$ and $x$ in $G=(V,E\setminus{\{u,v\}})$. Moreover, there exists a path of length at most $D-1$ between $v$ and $w$ in $G$. This true since $d_G(u,w) \leq D$. Since $x$ is a neighbor of $w$ it follows that the distance between every node $z\in R_v$ and $x$ is at most $D$. Thus removing the edge $\{u,v\}$ increases the diameter to at most $2D$ and agent $u$'s total distance cost by at most $2D|R_v|$.

We know that $G$ is in ($k$)NE in the deg($k$)NCG. Hence, buying the edge $\{u,v\}$ must be profitable for agent $u$, that is, $\deg_G(v)-1\leq 2D |R_v|$. Let $S(u)$ be the set of nodes to which agent $u$ bought a non-bridge edge. Summing up the inequalities for all nodes in $S(u)$ yields
$$\sum_{v\in S(u)}\left(\deg_G(v) -1\right) \leq 2D\sum_{v\in S(u)}|R_v| < 2nD,$$
where the last inequality holds since all sets $R_v$ are disjoint.
This implies that the total edge cost in any ($k$)NE network $G$ is at most $2n^2D+\Delta(n-1)$, which concludes the proof.
\end{proof}
From Corollary~\ref{thm:diameter} and~\ref{cor:diam4local} we know that the diameter of any NE in the degNCG and in the deg$4$NCG is at most $3$. Also, from the Lemma~\ref{lem:Edgecost}  we know that the social cost of any NE network $G$ is at most $\mathcal{O}(D(G))$ times the minimum possible social cost. This implies the following statement.
\begin{theorem}\label{thm:PoA}
The Price of Anarchy of the degNCG and the deg4NCG is in $\mathcal{O}(1)$.
\end{theorem}
\noindent A straightforward adaptation of Lemma~\ref{lem:Edgecost} together with Theorem~\ref{thm:generaldiam} yields:
\begin{corollary}\label{cor:generalconstantPoA}
The Price of Anarchy of variants of the degNCG where the price of any edge $\{u,v\}$ bought by agent $u$ is linear in $v$'s degree in $G$, is constant.  
\end{corollary}
\noindent Using Theorem~\ref{thm_diam_LNE} and Lemma~\ref{lem:Edgecost} yields the following statement. 
\begin{corollary}\label{cor_PoA_degLNCG}
The Price of Anarchy of the deg2NCG is in $\mathcal{O}(\sqrt{n})$. 
\end{corollary}
\noindent We conclude with analyzing the PoA in the deg($k$)AOG. The upper bound is trivially in $\mathcal{O}(n)$, the matching lower bound holds, since a clique is in ($k$)NE for the deg($k$)AOG for any $k$. 
\begin{observation}\label{pro:PoA-AOG}
The Price of Anarchy of deg(k)AOG is in $\Theta(n)$ for any $k$.
\end{observation}

\section{Dynamics}
In this section we consider the dynamic properties of the sequential version of the deg($k$)NCG and the deg($k$)AOG. The sequential version corresponds to an iterative process, called \emph{improving response dynamics (IRD)}, which starts with some initial strategy vector $\mathbf{s}$ and its corresponding initial network $G(\mathbf{s})$ and then agents are activated one at a time according to some activation scheme, e.g. a random or adversarially chosen move order or round-robin activation, and active agents are allowed to myopically update their current strategy. They will do so only if the new strategy yields strictly less cost than their current strategy. For the deg($2$)AOG we will also consider the \emph{best single edge dynamics}, which is a special case of the improving response dynamics, in which active agents buy the best possible single edge, if this strictly decreases their current cost.   

The most important dynamic property of a game is the finite improvement property (FIP)~\cite{MS96}, which states that any sequence of improving moves must be finite. The seminal paper~\cite{MS96} established that having the FIP is equivalent to being a generalized ordinal potential game. Thus, games having the FIP are guaranteed to converge to an equilibrium under improving move dynamics.

\subsection{Dynamics in the deg($k$)NCG}\label{sec:dyn_NCG}
We investigate the convergence properties of the deg($k$)NCG and prove that the deg($k$)NCG may not converge under improving move dynamics.

\begin{theorem}\label{thm:IRC}
 The deg(k)NCG does not have the FIP for any $k>1$, which implies that these games cannot have a generalized ordinal potential function.
\end{theorem}
\begin{proof}
 We prove the statement by providing an improving response cycle, which is a cyclic sequence of networks where neighboring networks differ only by the strategy of one agent and this strategy change is an improving response for the respective agent. See Fig.~\ref{fig:improvingmovecycle}.
 \begin{figure}
  \centering
  \includegraphics[width=12cm]{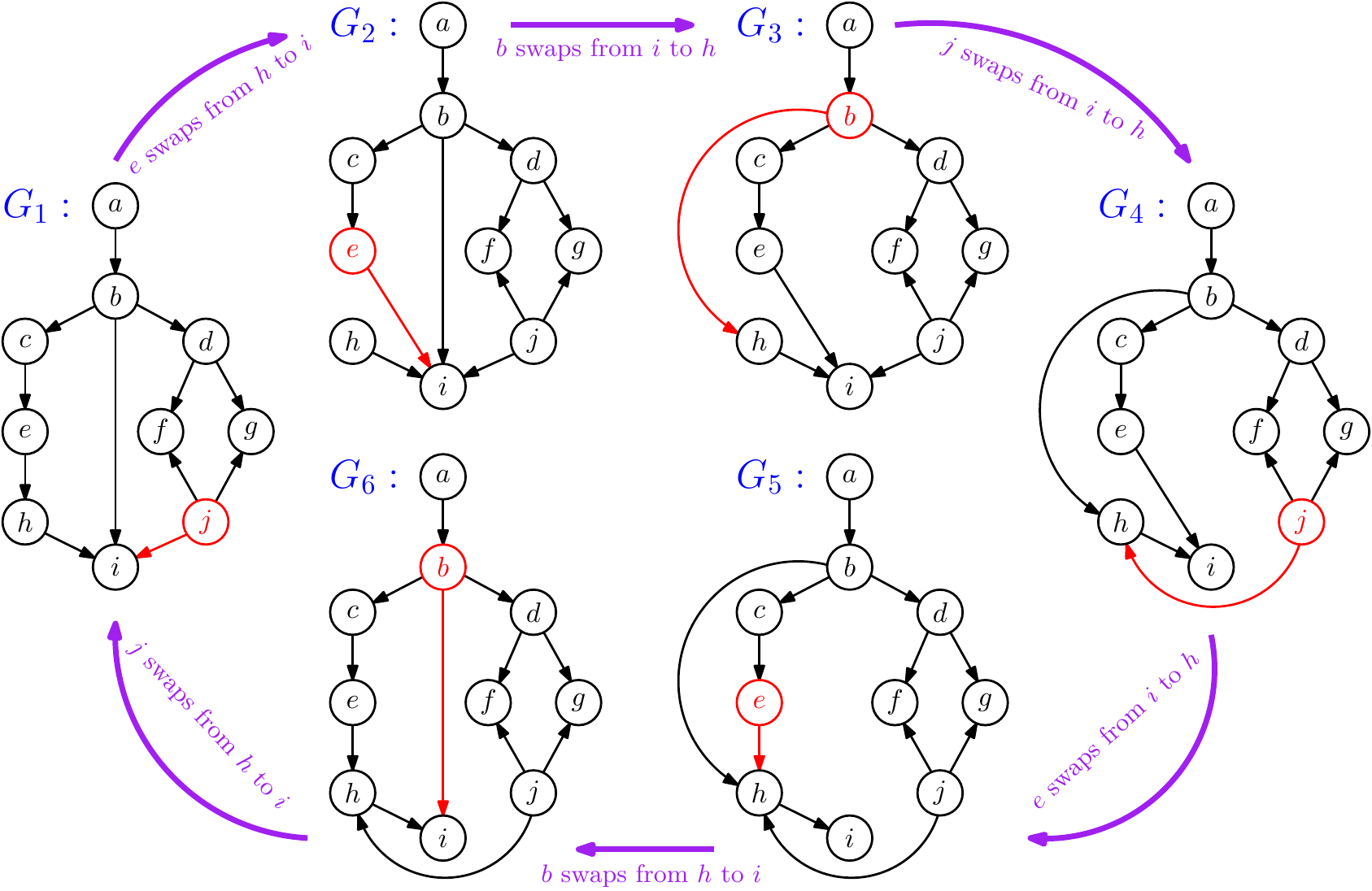}
  \caption{Example of an improving response cycle for the deg($k$)NCG.}
  \label{fig:improvingmovecycle}
 \end{figure}
The improving move cycle consists of six steps $G_1,\dots,G_6$ and the transition from step $G_i$ to $G_{i+1 \bmod 6}$ is an improving local move by some agent. Since all these improving moves are $2$-local, this proves the statement for both the deg$k$NCG and the degNCG for any $k>1$.

In network $G_1$ agent $e$ has edge cost of $1$ and distance cost of $23$ which yields a total cost of $24$. By buying the edge to node $i$ and removing the edge to node $h$ ($e$ ``swaps'' her edge from $h$ to $i$), agent $e$ can decrease her cost by $1$ since her edge cost in $G_2$ is $3$ and her distance cost is $20$. 

In network $G_2$ agent $b$ has edge cost of $6$ and distance cost of $14$ which yields total cost of $20$. By buying the edge to node $h$ and removing the edge to node $i$, agents $b$'s edges cost decreases to $4$ and her distance cost increases to $15$, which yields a total decrease of $1$.

In network $G_3$ agent $j$ has edge cost of $4$ and distance cost of $19$ giving total cost $23$. By performing the swap from $i$ to $h$, agent $j$'s edge cost does not change but her distance cost decreases by $1$.

To complete the improving move cycle, note that except for the edge ownership of edge $\{h,i\}$, network $G_4$ is isomorphic to $G_1$, network $G_5$ is isomorphic to $G_2$ and $G_6$ is isomorphic to $G_3$. Since neither agent $h$ nor agent $i$ are part of the improving move cycle, this implies that the described improving moves exist in these networks as well. That is, in network $G_4$ agent $e$ swaps her edge from $i$ back to $h$. In network $G_5$ agent $b$ swaps her edge from $h$ back to $i$ and in network $G_6$ agent $j$ swaps her edge from $h$ back to $i$.
 \end{proof}

\begin{remark}
 The presented improving response cycle in Fig.~\ref{fig:improvingmovecycle} is not a best response cycle for the deg($k$)NCG since in network $G_3$ agent $j$ has a strictly better local move: Buying the edge to agent $h$ and swapping her edge from $i$ to $e$.
\end{remark}

\subsection{Dynamics in the Add-Only Model}\label{sec:dyn_AOG}
We consider dynamics in the deg($k$)AOG. First of all, since agents can only add edges, the deg($k$)AOG trivially has the FIP, i.e. it is an ordinal potential game with the number of bought edges serving as a generalized ordinal potential function.

Since convergence is guaranteed, we focus on investigating the speed of convergence and the quality of the obtained networks. For the latter, Observation~\ref{pro:PoA-AOG} yields a devastating result. However, we contrast this for the deg$2$AOG by proving that if round-robin best single edge dynamics starting on a path as initial network are used, then the social cost is actually close to the best possible achievable social cost.

\begin{theorem}\label{thm:addonly-conver}
Let $P_n = \{v_1 \cdots v_n\}$ be the path of length $n$, with $v_1$ and $v_n$ as leaf nodes, as a initial graph for the deg(k)AOG:
\begin{enumerate}
\item If in any step the active agent is chosen uniformly at random then IRD in the deg(k)AOG converge in $\mathcal{O}(n^3)$ steps in expectation.
\item If in any step the active agent and her improving response is chosen adversarially then IRD in the deg(k)AOG converge in $\Theta(n^2)$ steps.
\item If round-robin best single edge dynamics are used in the deg2AOG, the process converges in at most $\mathcal{O}(n \log n)$ steps to a network with diameter $\mathcal{O}(1)$. 
\end{enumerate}
\end{theorem}
\begin{proof}
\begin{enumerate}
\item Consider the following procedure: At any time $t\geq 1$ let $G_t=(V,E_t)$ be the graph in the process where $G_0=P_n$, select a node $v\in V$  uniformly at random. If $v$ in $G_t$ has an improving move then $v$ adds a profitable edge. Otherwise it does nothing and the above step is repeated. Since the deg($k$)AOG is an ordinal potential game, we have that
at any time $t$, if $G_t$ is not in equilibrium, then there must be at least one agent having an improving response, i.e. who can buy at least one edge.

Consider the stochastic process $\{X_t\}_{t\geq 0}$, where $X_t= |E_t|$. Now if the graph $G_t$ is not in equilibrium, then the probability $Pr(X_t< X_{t+1})\geq \frac{1}{n}$ and  $Pr(X_t= X_{t+1})=1-Pr(X_t< X_{t+1})$. This implies that the expected number of steps until an improving response is played in the process $\{X_t\}_{t\geq 0}$ is dominated by the geometric random variable \textsc{$G\left(\frac{1}{n}\right)$}. In the process $\{X_t\}_{t\geq 0}$ the absorbing state corresponds to a ($k$)NE network. Since $n-1\leq X_t  \leq \frac{n(n-1)}{2}$, for any $t$, it follows that the expected number of steps needed to reach a ($k$)NE is at most $\mathcal{O}(n^3)$. 

\item Consider the adversarial scheme in Algorithm \ref{Alg:adversary} which is illustrated in Fig.~\ref{fig:advscheme}. First we prove that each addition of an edge induced by this scheme is an improving response. 
\begin{center}
\begin{minipage}{10cm}
\begin{algorithm}[H]
\caption{Adversarial Scheme for the deg($k$)AOG} 
\label{Alg:adversary}
\begin{algorithmic}[1]
\Require undirected path $P_n = (V,E)$ from $v_1$ to $v_n$
\For{$i := 1$ to $\lceil\frac{n}{2}\rceil-3$}
  \State activate $v_i$ and add the edge $\{v_i,v_{i+2}\}$ to $E$ 
	\For{$j= i-1$ to 1}
  	\State activate $v_j$ and add the edge $\{v_j,v_{i+2}\}$ to $E$
	\EndFor 
\EndFor
\For{$i=\lceil\frac{n}{2}\rceil+1$ to $n-2$}
	\State activate $v_{\lceil\frac{n}{2}\rceil-1}$ and add the edge  $\left\{v_{\lceil\frac{n}{2}\rceil-1},v_i\right\}$ 
	\EndFor
\State activate $v_n$ and add the edge $\{v_n, v_{n-2}\}$ to $E$
\If{the model under consideration is the deg$2$AOG} \State go to line 20 \EndIf
\State activate $v_{\lceil\frac{n}{2}\rceil}$ and add the edge $\left\{v_{\lceil\frac{n}{2}\rceil}, v_{n-1}\right\}$ 
\State $i := \lceil\frac{n}{2}\rceil+3$ 
\While{$i\leq n-5$}
	\State activate $v_{n-1}$ and add the edge $\{v_{n-1},v_i\}$ to $E$.
	\State $i := i+3$.
\EndWhile\\
\Return Graph $G=(V,E)$.
\end{algorithmic}
\end{algorithm}
\end{minipage}
\end{center}
Consider line 1 to 6 of the algorithm. At each step of the loop the degree of any node is at most $\lceil\frac{n}{2}\rceil$ and the distance improves to at least $\lceil\frac{n}{2}\rceil+1$. Thus every edge addition in the loop is an improving response for the activated node. 
Now in the next loop, line 7 to 9, only node $v_{\lceil\frac{n}{2}\rceil}$ is active and each of its additions is an improving response since all bought edges have a cost of~2 and the distance cost improvement is $n-i+1 \geq 3$. 
The edge added in line 10 is an improving response as it costs $3$ and the distance cost improvement is at least $n-2$.
\begin{figure}[!h]
  \centering
  \includegraphics[width=13cm]{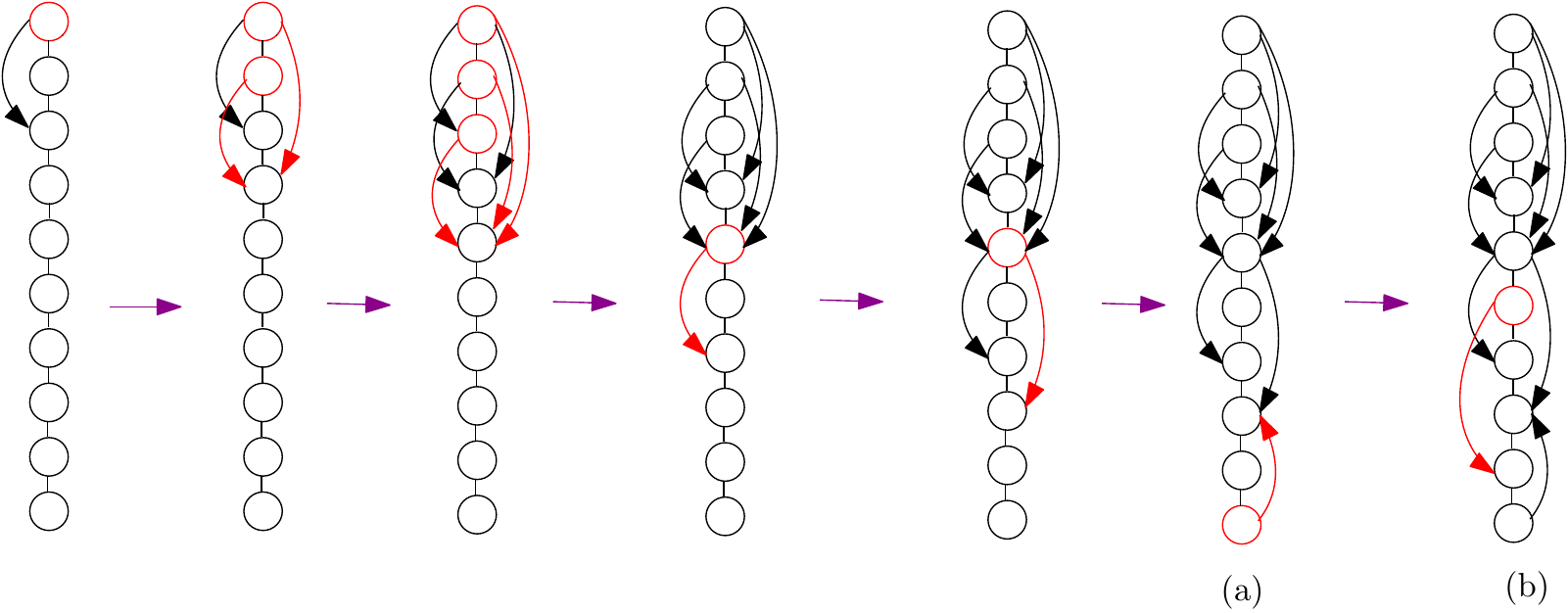}
  \caption{Example of an improving move sequence induced by Algorithm~\ref{Alg:adversary}. (a) is an 2NE in deg$2$AOG and (b) is $(k)$NE in deg$(k)$AOG, where $k\geq 3$.}
  \label{fig:advscheme}
 \end{figure}
 
Now after this step we note that the network is in $2$NE. Indeed, any two nodes from $\{v_1,\ldots, v_{n-2}\}$ are in distance at most 2 to each other and adding an edge to $v_{n-1}$ or $v_n$ improves the distance cost by only $1$. Note that adding the edges $\{v_{n-1},v_{n-3}\}$,$\{v_{n-1},v_{\lceil\frac{n}{2}\rceil-1}\}$, $\{v_{n-1}, v_{n-3}\}$ or $\{v_{n-1},v_{\lceil\frac{n}{2}\rceil-1}\}$ still may be an improvement. But buying an edge to $v_{\lceil\frac{n}{2}\rceil-1}$ costs $n-3$ and the distance cost improvement is at most $n-4$. At the same time adding $(v_{n-1},v_{n-3})$ shortcuts the distance to the nodes $v_{n-3}$ and $v_{n-4}$, thus this move improves the distance cost by 2 and has a cost of at least 3. Hence the graph is in $2$NE after line~10. 

Next, we prove that the remaining edge additions are improving responses in the degAOG. In line 14 the added edge yields an improvement for the active agent since it costs 2 and improves the distance to node $v_{n-1}$ by 2 and to node $v_{n}$ by 1. In the next loop, line 16 to 19, node $v_{n-1}$ is active and it buys an edge to $v_{\lceil \frac{n}{2}\rceil+3}$ and every third node starting from there. These additions are improving responses since the degree of each node $v_i$ is 3 and adding an edge to them improves $v_{n-1}$'s distance to $v_i$, $v_{i-1}$ and $v_{i-2}$ by 2, 1 and 1, respectively. After this step the diameter of the network is 3 and the degree of every node, except $v_n$, is at least 3, hence no agent can improve further, which implies that the network is in $(k)$NE, where $k\geq 3$ .

Note that already in the first loop, line 1 to 6, $\Theta(n^2)$ many improving responses are played.
\item Let $G = P_n$. We call an edge from $v_i$ to $v_j$ if $i<j$ a \emph{forward edge} otherwise, if $i>j$, a \emph{backward edge}. We consider round-robin activation in the order $v_1,v_2,\dots,v_n$. 

Starting with $v_1, v_2, \dots, v_{n-4}$ every agent $v_i$ can buy a forward edge to node $v_{i+2}$ since $deg_{G}(v_{i+2}) = 2$ and the improvement in distance cost is $n-(i+1)$. 
Another possibility for $v_i$ is to buy a backward edge to $v_{i-3}$ then it improves her distance to every second node in $v_1,\dots,v_{i-4}$ and to node $v_{i-3}$ by 1, i.e. distance cost improvement is at most $\frac{i-2}{2}$. 
At the same time buying a backward edge to $v_{i-4}$ costs $4$ and improves the distance cost by $i-4$. Hence, buying a forward edge is the best single edge addition for $v_i$ if $n-(i+1)-2\geq (i-4)-4$. Thus the best single edge addition of agent $v_{\lceil\frac{n+5}{2}\rceil}$ is an edge to $v_{\lceil\frac{n-3}{2}\rceil}$. 

Now agent $v_{\lceil\frac{n+5}{2}\rceil+1}$ can buy an edge to $v_{\lceil\frac{n+5}{2}\rceil+2}, v_{\lceil\frac{n+5}{2}\rceil-3}, v_{\lceil\frac{n+5}{2}\rceil-4}$ or to $v_{\lceil\frac{n+5}{2}\rceil-5}$. 
Observe that buying a backward edge to the node with the lowest possible index maximizes the distance cost improvement.
Adding the edge $\{v_\lceil\frac{n+5}{2}\rceil,v_{\lceil\frac{n+5}{2}\rceil-5}\}$  costs $5$ and it improves the distance towards all nodes $v_1,\ldots,v_{\lceil\frac{n+5}{2}\rceil-5}$ by 1. 
This yields a higher improvement than buying a forward edge since $\lceil\frac{n+5}{2}\rceil-4+1-5\geq n-(\lceil\frac{n+5}{2}\rceil+1)-2\Longleftrightarrow 2\cdot\lceil\frac{n+5}{2}\rceil\geq n+5$. Thus, adding $(v_{\lceil\frac{n+7}{2}\rceil},v_{\lceil\frac{n-3}{2}\rceil})$ is a best single edge purchase. 

For nodes $v_{\lceil\frac{n+7}{2}\rceil}$ and $v_{\lceil\frac{n+9}{2}\rceil}$  adding an edge to the same node $v_{\lceil\frac{n-3}{2}\rceil}$ saves $\lceil\frac{n-3}{2}\rceil+1$ and $\lceil\frac{n-3}{2}\rceil+2$ in distance cost, respectively, and the edge cost is $5$ and $6$, respectively. For all the other nodes with higher index buying an edge to node $v_{\lceil\frac{n-3}{2}\rceil}$ saves $\lceil\frac{n-3}{2}\rceil+3+(x-2)$ whereas edge costs are $x+4$, where $x$ is the number of nodes which already bought a backward edge to $v_{\lceil\frac{n-3}{2}\rceil}$. Since buying to a node with larger index or buying a forward edge is inferior, it follows that buying an edge to $v_{\lceil\frac{n-3}{2}\rceil}$ is a best single edge purchase for all agents with index larger than $\lceil\frac{n+5}{2}\rceil$.

So node $v_{\lceil\frac{n-3}{2}\rceil}$ becomes a high degree node with $deg_G(v_{\lceil\frac{n-3}{2}\rceil}) = n-\lceil\frac{n+5}{2}\rceil+1+4=\lfloor\frac{n+5}{2}\rfloor$. In the following we call nodes having a degree linear in $n$ a \emph{high degree node}.  
Consider the diameter of the network $D(G)$ after the first round. The length of the shortest path from $v_1$ to $v_{\lceil\frac{n-3}{2}\rceil}$ is halved by the forward edges and since every node between $v_{\lceil\frac{n+5}{2}\rceil}$ and $v_{n}$ buys an edge to $v_{\lceil\frac{n-3}{2}\rceil}$ the longest shortest path between the nodes $\{v_{\lceil\frac{n-3}{2}\rceil},\ldots,v_n\}$ is 3 so $D(G)$ decreases from $n-1$ to $\lceil\frac{n+9}{4}\rceil$. 

In the next round the best single improving edge of the first nodes $v_i$, where $i\leq \lceil\frac{n-11}{2}\rceil$ is an edge to $v_{i+4}$, since every edge costs $4$ and distance cost improvement is $n-(i+3)$. Buying a backward edge to $v_{i-8}$ has cost $6$ and improves the distance cost by $i-8+1$. 

Let the \emph{forward edge region} of any round be the set of nodes which have a best single edge purchase which buys a forward edge. In round $k$ no node up to node
$v_{\lceil\frac{n+2^k+2}{2}\rceil}$ will buy a backward edge, except when the forward edge would connect to a high degree node, since a forward edge from $v_{i}$ to $v_{i+2^k}$ costs $2k$ and improves the distance cost by $n-(i+2^k-1)$ while a backward edge from $v_{i}$ to $v_{i-2^{k+1}}$ costs $2k+2$ and the distance cost improvement is $i-2^{k+1}+1$.
It follows that the diameter of the forward edge region halves every round.

Now we consider the connection between high degree nodes during each round. In the second round the first backward edge is bought by agent $v_{\lceil\frac{n-11}{2}\rceil}$, since buying a forward edge to the high degree node $v_{\lceil\frac{n-3}{2}\rceil}$ is clearly too expensive. When we are in round $k$ a new high degree node $v_i$, with $i \geq \lceil\frac{n+21}{2}\rceil-3\cdot2^{k+1}$ to which other nodes buy backward edges will occur.

Since there are at most $2^k$ nodes between the first node who starts to buy a backward edge and the high degree node from the last round and to the next high degree nodes there are at most $3\cdot 2^{k-1}$, $3\cdot 2^{k-2}$ etc. nodes which lie in-between, there can be at most $\mathcal{O}(\log n)$ high degree nodes. Consider the high degree node with lowest index in round $k\geq 2$ and consider the step where the high degree node with the next higher index is active. At this point, all nodes between the two high degree nodes have bought a backward edge to the high degree node with lowest index, which implies that the high degree node with the second lowest index can buy a backward edge as well. Since not all nodes with lower index than the high degree node with lowest index have bought a forward edge to the high degree node with lowest index, it follows that the best single edge of the high degree node with the second lowest index connects to the high degree node with lowest index. Analogously, all high degree nodes with higher index will also buy an edge to the high degree node with lowest index. Thus, all high degree nodes after any round will form a clique. It follows that the distance of any node with higher index than the current high degree node with lowest index to the latter node is at most $3$.   

\begin{figure}[!h]
  \centering
  \includegraphics[angle = 90, width=\textwidth]{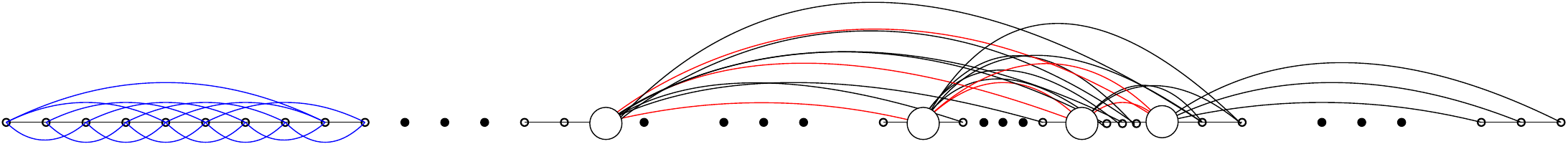}
  \caption{A sketch of the graph $G$ in round $k$. The forward edge region, where the best single edge addition is to buy a forward edge, is shown in blue. The large nodes show the high degree nodes and the red edges belong to the clique formed by the high degree nodes.}
  \label{fig:roundrobinsketch}
 \end{figure}

Since the diameter of the forward edge region halves in every round there can be at most $\mathcal{O}(\log n)$ rounds until the diameter of the network is in $\mathcal{O}(1)$. At this point all nodes with lower index than the lowest index of the clique nodes can possibly buy an improving forward edge to all clique nodes but no improving backward edge can be bought by any agent. Thus, there can be $\mathcal{O}(\log n)$ additional rounds and the total number of additional edges is in $\mathcal{O}(n\log n)$.  

In total the dynamics needed $\mathcal{O}(\log n)$ rounds and $O(n \log n)$ best single edge improvements are made.
\end{enumerate}
\end{proof}
\noindent We contrast the upper bounds by showing that convergence in $\mathcal{O}(n)$ many improving responses is possible.
\begin{theorem}\label{thm:converge-best}
Let $P_n$ be the initial network then there exists a sequence of improving responses which takes
\begin{enumerate}
\item  $n-2+\frac{n-7}{3}$ steps to obtain a NE network in the degAOG;
\item $n-1$ steps to obtain a 2NE network in the deg2AOG.
\end{enumerate}
\end{theorem}
\begin{proof}
\begin{enumerate}
\item Consider the path $P_n = v_1,\cdots,v_n$. We activate $v_1$ and sequentially buy the edges $\{v_1,v_i\}$ for each $3\leq i\leq n-2$. Afterwards we activate $v_{n-1}$ and sequentially buy the edges $\{v_{n-1},v_{i+3}\}$ with $0\leq i\leq n-8$. Each edge addition is an improving move since $v_1$ pays $2$ to improve the distance cost by at least 3 and $v_n$ pays 3 and improves her distance cost by 4 to the nodes $v_{i+2},$ $v_{i+3}$ and $v_{i+4}$. After that no edge addition is possible as each edge costs at least 3 whereas the possible distance cost improvement is at most 3. Thus the network is in NE after $n-2+\frac{n-7}{3}=\Theta(n)$ many steps. 

\item Consider the following activation scheme in which node $v_1$ buys an edge to each node $v_i$ where $3\leq i\leq n-2$. Every single edge purchase is an improving response because the edge cost is $2$ and the distance cost improvement is at least $3$. Next, we activate $v_n$ and add the edge $\{v_n,v_{n-2}\}$. This is improving for $v_n$, since the edge costs $3$ and it decreases the distance cost by $n-2$. 

All nodes $\{v_1, \dots, v_{n-2}\}$ are in distance at most 2 towards each other, thus there are no possible further improving edge additions between these agents. 
Buying an edge to $v_{n-1}$ and $v_{n}$ costs 2 and 1, respectively, whereas the distance cost improvement is also $2$ and $1$, respectively. 
The nodes $v_{n-1}$ and $v_{n}$ do not want to buy an edge to $v_1$ as it costs $n-3$ and the distance cost improvement is $n-4$ and there is no other local improving response for both nodes as each node in their 2-neighborhood costs 3 and the distance cost improvement is just 1.     
Therefore the network is in $2$NE after $n-1$ many steps.  
\end{enumerate}
\end{proof}
\noindent Finally, we investigate the quality of the ($2$)NE networks which can be obtained by improving move dynamics starting from the path $P_n$. 
For this, we introduce a measure which is similar to the Price of Anarchy. Let $G_0$ be any initial connected network and let $Z(G_0)$ be the set of networks which can be obtained via improving response dynamics in the deg($2$)AOG. Let $Best(G_0) \in Z(G_0)$ be the reachable network with minimum social cost among all networks in $Z(G_0)$. We can now measure the quality of any network $G \in Z(G_0)$ by investigating the ratio $\rho(G,G_0) = \frac{cost(G)}{cost(Best(G_0))}$. 
\begin{theorem}\label{thm:rho-dynamics}
~
\begin{enumerate}
\item Let $G$ be any network in $Z(G_0)$ then $\rho(G,G_0) \in \mathcal{O}(n).$ 
\item There is a network $G \in Z(P_n)$ for the deg(2)AOG with $\rho(G,P_n) \in \Theta(n)$. 
\item Let $G^*$ be the network obtained by the round-robin best single edge dynamics in the deg2AOG, then we have $\rho(G^*,P_n) \in \mathcal{O}(\log n)$.
\end{enumerate}
\end{theorem}
\begin{proof}
\begin{enumerate}
\item To give the upper-bound on $\rho(G,G_0)$ we give upper-bound on the social cost of $G$. Let $D$ be the diameter of $G$ then the total distance cost of all agents is at most $D n^2 \leq n^3$. The total edge cost of all agents is also is in $\mathcal{O}(n^3)$. The social of $Best(G_0)$ is at least $n(n-1)$, since the diameter of $Best(G_0)$ is at least~$1$. Hence,
$\rho(G,G_0)\in {\mathcal{O}(n)}$.
\item The upper-bound follows from (1). The matching lower bound follows from the example network $G$ given in the second part of the proof of Theorem~\ref{thm:addonly-conver}. In $G$ there are $\Theta(n^2)$ many edges which all have cost in $\Theta(n)$ which yields social cost in $\Theta(n^3)$.
Now consider the activation scheme in the proof of Theorem~\ref{thm:converge-best}. In the constructed equilibrium network $G'$ the total edge cost is in $\Theta(n)$. Thus the social cost of $G'$ is in $\Theta(n^2)$, which is an upper bound on the social cost of $Best(P_n)$. Hence, the $\rho(G,P_n) \in \Omega(n)$. 
\item  We prove the upper bound on $\rho(G^*,P_n)$ by giving the upper bound on the social cost of $G^*$ and the lower bound on $Best(P_n)$.

From the third part of Theorem~\ref{thm:addonly-conver} we know that the total number of edges and diameter of $G^*$, is in $\mathcal{O}(n\log n)$ and $\mathcal{O}(1),$ respectively. Since the diameter of the $G^*$ is in $\mathcal{O}(1)$ the total distance cost can be upper bounded by $\mathcal{O}(n^2)$. The total edge cost of $G^*$ is upper bounded by $\mathcal{O}(n^2 \log n)$, since there are $\mathcal{O}(\log n)$ many high degree nodes with degree $\Theta(n)$ and almost all edges are bought towards these high degree nodes and thus each have cost in $\Theta(n)$. Hence, $G^*$ has a social cost in $\mathcal{O}(n^2\log n)$. 

On the other hand, a trivial lower bound on the social cost of $Best(P_n)$ is $n(n-1)$, which then yields $\rho(G^*,P_n) \in \mathcal{O}(\log n)$.
\end{enumerate}
\end{proof}

\section{Conclusion}
We have introduced natural variants of the well-known NCG by Fabrikant et al.~\cite{Fab03}, which have the distinctive features that they are parameter-free and at the same time incorporate non-uniform edge costs. Besides proving that computing a best response is \NP-hard and that improving response dynamics may never converge to an equilibrium, we have also established that the degNCG has a constant Price of Anarchy. This strong statement holds whenever the edge price is any linear function of the degree of the non-owner endpoint of the edge or if agents are allowed to buy edges to nodes in their $4$-neighborhood. For the version which includes stronger locality, i.e. the deg$2$NCG, we have shown that the PoA is in $\mathcal{O}(\sqrt{n})$ and, as a contrast, for the add-only version the PoA is in $\Theta(n)$. We also demonstrate how to circumvent the latter negative result by using suitable activation schemes on a sparse initial network.  

Studying the bilateral version of our model, where both endpoints of the edge have to agree and pay proportionally to the degree of the other endpoint for establishing an edge, is an obvious future research direction. For this version, we have already established that most of our proofs can be easily adapted, which implies that our results, with minor modifications, still hold. Another interesting extension would be to consider an edge price function which depends on the degree of \emph{both} involved nodes. This could be set up such that edges between nodes of similar degree are cheap and edges become expensive when the degree of both nodes differs greatly.

\bibliographystyle{abbrv}
\bibliography{degpriceNCG}

\end{document}